\documentclass[11pt]{article}

\usepackage{amsmath,amssymb,amsthm}
\usepackage{geometry}
\usepackage{hyperref}

\newtheorem{theorem}{Theorem}[section]
\newtheorem{proposition}[theorem]{Proposition}
\newtheorem{remark}[theorem]{Remark}
\theoremstyle{definition}
\newtheorem{assumption}{Assumption}[section]

\newtheorem{lemma}{Lemma}[section]

\title{Analytic Regularity and Approximation Limits of Coefficient-Constrained Shallow Networks}

\author{Jean-Gabriel Attali\\
De Vinci Higher Education, De Vinci Research Center, Paris, France \\
\texttt{jean-gabriel.attali@devinci.fr}}

\date{}

\begin{document}
\maketitle

\begin{abstract}
We study approximation limits of single-hidden-layer neural networks with analytic
activation functions under global coefficient constraints.
Under uniform $\ell^1$ bounds, or more generally sub-exponential growth of the
coefficients, we show that such networks generate model classes with strong quantitative
regularity, leading to uniform analyticity of the realized functions.
As a consequence, up to an exponentially small residual term, the error of best network
approximation on generic target functions is bounded from below by the error of best
polynomial approximation.
In particular, networks with analytic
activation functions with controlled coefficients cannot outperform
classical polynomial approximation rates on non-analytic targets.
The underlying rigidity phenomenon extends to smoother, non-analytic activations
satisfying Gevrey-type regularity assumptions, yielding sub-exponential variants of the
approximation barrier.
The analysis is entirely deterministic and relies on a comparison argument combined with
classical Bernstein-type estimates; extensions to higher dimensions are also discussed.
\end{abstract}



\section{Introduction}

Single-hidden-layer neural networks are classical nonlinear approximation tools whose theoretical properties have been extensively studied over the past decades. Early results established their universal approximation capabilities under mild assumptions on the activation function. More refined analyses later showed that, under additional structural assumptions on the target function, such networks may achieve remarkably fast approximation rates.
A prominent example is provided by the theory initiated by Barron \cite{Barron1993}, and earlier related contributions such as Attali and Pag`es \cite{AttaliPages1995,AttaliPages1997}, which show that neural networks can overcome the curse of dimensionality when the target function belongs to a suitable analytic or spectral class. These results crucially rely on global constraints on the network coefficients and on strong regularity properties of the target function itself.
In parallel, classical approximation theory provides sharp minimax bounds for generic smoothness classes such as Lipschitz or Sobolev spaces. These bounds show that, in the absence of additional structure, no approximation method can outperform polynomial rates \cite{DeVoreLorentz1993,DeVore1998}. This fundamental limitation naturally raises the following question:
\emph{can analyticity of the activation function alone improve approximation rates on generic target functions?}
Before addressing this question, it is important to note that analyticity-based approximation properties depend not only on the activation function but also on the effective control of the network parameters. In particular, uniform analyticity on a fixed complex neighborhood is obtained when the inner parameters remain suitably bounded; more generally, increasing inner parameters lead to shrinking analytic neighborhoods and exponentially small residual terms, as discussed later.
In this work, we provide a negative answer to the above question. We show that analyticity of the activation function, even when combined with uniform $\ell^1$ bounds on the network coefficients, does not suffice to overcome the classical polynomial approximation barrier on non-analytic target functions. More precisely, we consider single-hidden-layer networks with real-analytic activation functions and uniformly bounded coefficient sums, and we show that the error of best network approximation is bounded from below by the error of best polynomial approximation, up to an exponentially small term. This result should be understood as a rigidity phenomenon for analytic model classes, rather than as a minimax lower bound.
Our results reveal a structural limitation of networks with analytic
activation functions. While analyticity plays a central role in positive approximation results for structured target classes, it also imposes intrinsic constraints on the class of realizable functions. As a consequence, networks with analytic
activation functions cannot adapt to non-analytic features of generic target functions and cannot achieve approximation rates faster than those dictated by classical polynomial approximation theory.
The proof is entirely deterministic and relies on a comparison between best network approximation and best polynomial approximation, combined with classical Bernstein-type estimates for analytic functions. In contrast with probabilistic constructions or random sampling arguments, our approach provides a transparent explanation of why analyticity of the activation function alone cannot bypass minimax approximation barriers. Although the main results are presented for analytic activation functions, the underlying rigidity mechanism extends to smoother, non-analytic activations under quantitative regularity assumptions, such as Gevrey smoothness, leading to sub-exponential variants of the approximation barrier.

\begin{remark}

Throughout the paper, we restrict attention to functions defined on the interval $[-1,1]$.
This choice is made for notational convenience only: any compact interval can be reduced to
this setting by an affine change of variables, without affecting approximation rates or
analyticity properties.
Working on $[-1,1]$ allows us to use standard Bernstein ellipses and simplifies the presentation.
\end{remark}

The paper is organized as follows.
Section~2 introduces the class of analytic activation networks considered throughout
the paper and establishes their basic analytic properties.
Section~3 recalls classical results on polynomial approximation of analytic functions.
The main lower bound is stated and proved in Section~4.
Finally, Section~5 discusses the implications of our results and their relation to
existing approximation theories for neural networks.

\section{Analytic Activation Networks}

In this section, we introduce the class of neural networks considered throughout the paper and establish their basic analytic properties. The framework is entirely deterministic and relies on analyticity of the activation function together with global bounds on the network coefficients.
A key point is that the approximation properties derived below hold in a regime where the analytic regularity of the resulting model class can be quantitatively controlled. In particular, uniform analyticity on a fixed complex neighborhood is obtained under suitable control of the inner parameters; more general parameter growth leads to shrinking analytic neighborhoods, which will be accounted for explicitly in the subsequent estimates.

\subsection{Model definition}

Let $\varphi:\mathbb{R}\to\mathbb{R}$ be a real-analytic function.
We consider single-hidden-layer neural networks of the form
\[
g(x)=\sum_{k=1}^m \lambda_k\,\varphi(\alpha_k x),
\qquad x\in[-1,1],
\]
where $m\ge1$ and $\lambda_k,\alpha_k\in\mathbb{R}$.
For a fixed constant $C>0$, we denote by $\mathcal{N}_m(C)$ the class of network functions
whose output weights satisfy the uniform $\ell^1$ constraint
\[
\sum_{k=1}^m |\lambda_k|\le C.
\]

We omit bias terms for simplicity.
Their inclusion in expressions of the form $\varphi(\alpha x + b)$ would not affect any of the
arguments below, as long as the biases remain uniformly bounded.

The analysis relies on the fact that the resulting model class enjoys a controlled analytic
structure.
More precisely, we assume that for each $m$ there exists a complex neighborhood
$U_m\subset\mathbb{C}$ containing the interval $[-1,1]$ and a constant $C_m>0$ such that every
function $g\in\mathcal{N}_m(C)$ admits a holomorphic extension to $U_m$ satisfying
\[
\sup_{z\in U_m}|g(z)|\le C_m .
\]
This assumption captures the uniform analyticity properties induced by the activation function
together with the imposed parameter constraints, and constitutes the only structural hypothesis
used in the comparison arguments developed below.

The $\ell^1$ constraint on the output weights plays a central role in controlling the complexity
of the model.
Such constraints are classical in approximation theory and arise naturally in several contexts,
including Barron-type representations and earlier works on multilayer perceptrons
\cite{Barron1993,AttaliPages1997}.
Throughout the paper, the constraint is imposed deterministically and uniformly with respect to
the network width~$m$.

\subsection{Analytic extension and uniform bounds}

A key consequence of analyticity of the activation function is that it imposes strong
regularity constraints on the functions generated by the network.

\begin{assumption}[Analytic extension]\label{ass:analytic-extension}

Fix $\rho>1$ and $L\ge 1$.
The activation function $\varphi$ admits a holomorphic extension to an open neighborhood of
the dilated Bernstein ellipse $L E_\rho := \{Lz : z\in E_\rho\}$, and we set
\[
M_{\rho,L}(\varphi):=\sup_{z\in L E_\rho}|\varphi(z)|<\infty.
\]
\end{assumption}

Under this assumption, the analytic structure of network outputs can be controlled
uniformly.

\begin{proposition}[Uniform analytic control]\label{prop2}
Assume $\varphi$ satisfies Assumption~2.1.
Let $L\ge 1$ and consider network functions
\[
g(x)=\sum_{k=1}^m \lambda_k\,\varphi(\alpha_k x),\qquad x\in[-1,1],
\]
with $\sum_{k=1}^m|\lambda_k|\le C$ and $|\alpha_k|\le L$ for all $k$.
Then $g$ admits a holomorphic extension to $E_\rho$ and satisfies
\[
\sup_{z\in E_\rho}|g(z)|\le C\,M_{\rho,L}(\varphi).
\]
\end{proposition}

\begin{proof}
Let $g(x)=\sum_{k=1}^m \lambda_k\,\varphi(\alpha_k x)$.
Fix $z\in E_\rho$. Since $|\alpha_k|\le L$, we have $\alpha_k z\in L E_\rho$.
By Assumption~2.1, $\varphi$ is holomorphic on a neighborhood of $L E_\rho$; hence
$z\mapsto \varphi(\alpha_k z)$ is holomorphic on $E_\rho$, and
\[
|\varphi(\alpha_k z)|\le \sup_{w\in L E_\rho}|\varphi(w)|=M_{\rho,L}(\varphi).
\]
Therefore, by linearity and the $\ell^1$ bound on the coefficients,
\[
\sup_{z\in E_\rho}|g(z)|
\le
\sum_{k=1}^m |\lambda_k|\,
\sup_{z\in E_\rho}|\varphi(\alpha_k z)|
\le
\Big(\sum_{k=1}^m|\lambda_k|\Big)\,M_{\rho,L}(\varphi)
\le
C\,M_{\rho,L}(\varphi),
\]
which proves the claim.
\end{proof}

\subsection{Smooth non-analytic activations: a Gevrey relaxation}

The analyticity assumption imposed in the previous sections can be relaxed to a broader
class of infinitely differentiable activation functions, provided that their smoothness
is controlled in a quantitative manner.
In this subsection, we introduce a canonical relaxation based on Gevrey regularity,
which allows one to retain a weakened but still effective form of the polynomial
approximation barrier established for analytic activations.

\paragraph{Gevrey regularity.}
Let $s \ge 1$. A function $\varphi \in C^\infty(\mathbb{R})$ is said to belong to the
Gevrey class $G^s(\mathbb{R})$ if there exist constants $C_\varphi>0$ and $R_\varphi>0$ such that
\begin{equation}
\sup_{t\in\mathbb{R}} |\varphi^{(n)}(t)| \le C_\varphi\, R_\varphi^{\,n}\,(n!)^{s},
\qquad \forall n\ge 0.
\label{eq:gevrey_phi}
\end{equation}
The case $s=1$ corresponds, up to constants, to real-analytic functions, while $s>1$
describes a large family of $C^\infty$ but non-analytic functions, including standard
compactly supported mollifiers and smooth activations obtained by regularizing
piecewise linear functions.

Gevrey classes provide a classical quantitative relaxation of analyticity, interpolating between
real-analytic and general $C^\infty$ regularity. They play a central role in approximation theory
and the study of sub-exponential approximation rates; see, e.g., Rodino~\cite{Rodino1993} or
DeVore and Lorentz~\cite{DeVoreLorentz1993}.

\paragraph{Uniform Gevrey control of network outputs.}
We consider single-hidden-layer networks of the form
\begin{equation}
g_m(x) = \sum_{k=1}^m \lambda_k\,\varphi(\alpha_k x),
\qquad x\in[-1,1],
\label{eq:gevrey_network}
\end{equation}
under the coefficient constraints
\begin{equation}
\sum_{k=1}^m |\lambda_k| \le B_m,
\qquad |\alpha_k| \le L_m .
\label{eq:coeff_constraints}
\end{equation}

\begin{proposition}[Uniform Gevrey control] \label{eq:gevrey_network_derivatives}
Assume that $\varphi\in G^s(\mathbb{R})$ satisfies \eqref{eq:gevrey_phi}.
Then every function $g_m$ of the form \eqref{eq:gevrey_network} satisfies
\begin{equation}
\| g_m^{(n)} \|_{L^\infty([-1,1])}
\le C_\varphi\, B_m\, (R_\varphi L_m)^n\, (n!)^{s},
\qquad \forall n\ge 0.
\end{equation}
In particular, the family of network outputs is uniformly bounded in the Gevrey class
$G^s([-1,1])$, with constants depending only on $B_m$ and $L_m$.
\end{proposition}

\begin{proof}
Differentiating \eqref{eq:gevrey_network} $n$ times yields
\[
g_m^{(n)}(x)
= \sum_{k=1}^m \lambda_k\, \alpha_k^{\,n}\, \varphi^{(n)}(\alpha_k x).
\]
Using the triangle inequality, the coefficient constraints
\eqref{eq:coeff_constraints}, and the Gevrey bound \eqref{eq:gevrey_phi}, we obtain
\[
|g_m^{(n)}(x)|
\le \sum_{k=1}^m |\lambda_k|\, |\alpha_k|^n
\sup_{t\in\mathbb{R}} |\varphi^{(n)}(t)|
\le C_\varphi\, B_m\, (R_\varphi L_m)^n\, (n!)^{s},
\]
uniformly for $x\in[-1,1]$, which proves the claim.
\end{proof}

\subsection{Discussion and relation to analytic and Gevrey frameworks}

Proposition~\ref{prop2} and Proposition~\ref{eq:gevrey_network_derivatives} show that, under global $\ell^1$ constraints,
single-hidden-layer networks generate model classes with strong quantitative
regularity, ranging from uniform analyticity to Gevrey smoothness.
This property is central to the comparison arguments developed in the subsequent
sections.

Related analytic regularity phenomena have already been observed in earlier works on
neural network approximation.
In particular, Attali and Pag\`es \cite{AttaliPages1997} studied approximation by
multilayer perceptrons with analytic activation functions and emphasized the strong
regularity induced by analyticity.
Barron-type results \cite{Barron1993} also rely on global coefficient constraints, but
their conclusions concern approximation on structured target classes rather than
generic functions.

The present work adopts a different perspective.
Rather than exploiting analyticity to derive positive approximation rates, we use it
to identify intrinsic limitations of analytic activation networks on generic target
functions.
The uniform analytic bounds established above will serve as the starting point for the
polynomial comparison arguments developed in Sections~3 and~4.

\paragraph{Historical context.}
Explicit low-dimensional instances of the rigidity induced by analyticity were already observed
in early works by Attali and Pag\`es \cite{AttaliPages1995,AttaliPages1997}.
In particular, these authors provided an explicit bivariate proof of the convergence of partial
derivatives for Bernstein-type approximants, placing the result in an appendix as a classical
technical ingredient in their analysis of multilayer perceptrons.
From a modern perspective, this argument can be viewed as an early manifestation of the
general rigidity phenomena associated with globally parametrized analytic approximation schemes.

\section{Polynomial Approximation of Analytic Functions}

In this section, we recall classical results on the approximation of analytic functions
by algebraic polynomials.
These results form a cornerstone of approximation theory and will be used as a key
ingredient in the proof of the main lower bound.
Standard references include the works of Bernstein, Timan, DeVore and Lorentz, and
more recent expositions such as \cite{Bernstein1912,Timan1963,DeVoreLorentz1993,Trefethen2013}.

\subsection{Bernstein ellipses and analytic norms}

Let $K=[-1,1]$.
For $\rho>1$, denote by $E_\rho$ the Bernstein ellipse with foci at $\pm1$ and parameter
$\rho$, defined as the image of the circle $\{w\in\mathbb{C}:\ |w|=\rho\}$ under the
Joukowski map
\[
w \mapsto \tfrac12\bigl(w+w^{-1}\bigr).
\]
The ellipse $E_\rho$ is a compact subset of $\mathbb{C}$ containing $[-1,1]$.

For a function $h$ holomorphic on $E_\rho$, we define the associated analytic norm by
\[
M_\rho(h)=\sup_{z\in E_\rho}|h(z)|.
\]
This quantity provides a convenient measure of the strength of analyticity of $h$ in a
neighborhood of the real interval.

\subsection{Best polynomial approximation}

For a continuous function $h\in C([-1,1])$, we denote by
\[
E_m(h)=\inf\bigl\{\|h-p\|_\infty:\ p \text{ is a polynomial of degree at most } m\bigr\}
\]
the error of best uniform polynomial approximation of degree at most $m$.

If $h$ is merely continuous, the decay of $E_m(h)$ is typically polynomial and governed
by the smoothness of $h$.
In contrast, analyticity of $h$ implies exponentially fast decay of the approximation
error.

\begin{proposition}[Bernstein-type inequality]\label{prop:bernstein}
Let $\rho>1$ and let $h$ be holomorphic on $E_\rho$.
Then there exists a constant $A_\rho>0$, depending only on $\rho$, such that
\[
E_m(h)\le A_\rho\,M_\rho(h)\,\rho^{-m},
\qquad m\ge1.
\]
\end{proposition}

This result is classical and can be found in many textbooks on approximation theory;
see for instance \cite{Bernstein1912,Timan1963,DeVoreLorentz1993,Trefethen2013}.

\subsection{Consequences for analytic model classes}

Proposition~\ref{prop:bernstein} shows that any family of functions admitting a uniform
analytic extension to a fixed Bernstein ellipse necessarily enjoys exponentially fast
polynomial approximation.

More precisely, if a class $\mathcal{F}\subset C([-1,1])$ satisfies
\[
\sup_{h\in\mathcal{F}} M_\rho(h)<\infty
\]
for some $\rho>1$, then every function in $\mathcal{F}$ can be approximated uniformly by
polynomials at an exponential rate.

As shown in Section~2, networks with analytic
activation functions with uniformly $\ell^1$-bounded
coefficients generate families of functions that satisfy exactly such uniform analytic
bounds.
This observation provides the crucial link between neural network approximation and
classical polynomial approximation and forms the basis of the comparison argument
developed in the next section. From the viewpoint of classical approximation theory, networks with analytic
activation functions with
global coefficient constraints form a precompact family in C([-1,1]).
This observation explains why approximation rates beyond those achieved by polynomial
methods cannot be expected on generic target functions.

\paragraph{Remark (Gevrey extension).}
Although the results of this section are stated for analytic functions,
classical approximation theory provides analogous estimates for functions
with Gevrey regularity.
In that case, exponential polynomial approximation rates are replaced by
sub-exponential bounds depending on the Gevrey index.
Since our arguments rely only on quantitative polynomial approximation
estimates, the comparison results of Section~4 extend verbatim to the
Gevrey setting introduced in Section~2.3.

\section{Polynomial approximation barrier}

This section establishes a lower bound showing that single-hidden-layer neural networks with
analytic or Gevrey activation functions and globally controlled coefficients remain confined to
classical polynomial approximation regimes on generic target functions. The argument relies on
quantitative regularity of the model class and a direct comparison with best polynomial
approximation.

\subsection{Bernstein ellipses and polynomial approximation}

For $\rho>1$, let $E_\rho$ denote the Bernstein ellipse with foci $\pm1$. For a function $h$
holomorphic on $E_\rho$, define
\[
M_\rho(h) := \sup_{z\in E_\rho} |h(z)|.
\]
The error of best uniform polynomial approximation of degree at most $m$ is
\[
E_m(h) := \inf_{\deg p\le m} \|h-p\|_{L^\infty([-1,1])}.
\]

\begin{proposition}[Bernstein inequality]
\label{prop:bernstein}
For every $\rho>1$ there exists a constant $C(\rho)>0$ such that, for all functions $h$
holomorphic on $E_\rho$,
\[
E_m(h) \le C(\rho)\, M_\rho(h)\, \rho^{-m},
\qquad m\ge1.
\]
\end{proposition}

\subsection{Analytic activation functions}

Let $\varphi:\mathbb R\to\mathbb R$ be a real-analytic activation function.
For $m\ge1$, consider single-hidden-layer networks of the form
\[
g_m(x) = \sum_{k=1}^m \lambda_k \varphi(\alpha_k x), \qquad x\in[-1,1],
\]
under the constraints
\[
\sum_{k=1}^m |\lambda_k| \le B_m,
\qquad
|\alpha_k|\le L_m.
\]

\begin{assumption}[Analytic extension on a Bernstein ellipse]
\label{ass:analytic-ellipse}
There exist $\rho>1$ and $L\ge1$ such that $\varphi$ admits a holomorphic extension to a
neighborhood of the dilated ellipse $L E_\rho := \{Lz:\ z\in E_\rho\}$, and
\[
M_{\rho,L}(\varphi) := \sup_{w\in L E_\rho} |\varphi(w)| < \infty.
\]
\end{assumption}

\begin{lemma}[Uniform analytic control]
\label{lem:analytic-control}
Assume that $\varphi$ satisfies Assumption~\ref{ass:analytic-ellipse} and that $L_m\le L$.
Then $g_m$ admits a holomorphic extension to $E_\rho$ and
\[
M_\rho(g_m) \le B_m\, M_{\rho,L}(\varphi).
\]
\end{lemma}

\begin{proof}
For $z\in E_\rho$ and $1\le k\le m$, the bound $|\alpha_k|\le L$ implies $\alpha_k z\in L E_\rho$.
Thus
\[
|\varphi(\alpha_k z)| \le M_{\rho,L}(\varphi),
\]
and
\[
|g_m(z)|
\le \sum_{k=1}^m |\lambda_k|\,|\varphi(\alpha_k z)|
\le B_m\, M_{\rho,L}(\varphi).
\]
\end{proof}

\begin{proposition}[Polynomial approximation of analytic network outputs]
\label{prop:analytic-poly}
Under the assumptions of Lemma~\ref{lem:analytic-control}, there exists $C(\rho)>0$ such that
\[
E_m(g_m) \le C(\rho)\, B_m\, M_{\rho,L}(\varphi)\, \rho^{-m}.
\]
\end{proposition}

\begin{theorem}[Polynomial approximation barrier: analytic case]
\label{thm:analytic-barrier}
Let $\varphi$ satisfy Assumption~\ref{ass:analytic-ellipse}. Then for all
$f\in C([-1,1])$ and all $m\ge1$,
\[
\inf_{g_m} \|f-g_m\|_{L^\infty([-1,1])}
\ge
E_m(f)
-
C(\rho)\, B_m\, M_{\rho,L}(\varphi)\, \rho^{-m},
\]
where the infimum is taken over all networks satisfying the above constraints.
\end{theorem}

\begin{proof}
Let $p_m$ be a best polynomial approximant of degree at most $m$ to $g_m$. Then
\[
\|f-g_m\|_\infty
\ge \|f-p_m\|_\infty - \|g_m-p_m\|_\infty
\ge E_m(f) - E_m(g_m).
\]
Apply Proposition~\ref{prop:analytic-poly}.
\end{proof}

\begin{remark}[Optimized scaling]
If $\varphi$ is holomorphic and bounded only on a strip $|\Im z|<\delta$, then $g_m$ is
holomorphic on $|\Im z|<\delta/L_m$. Choosing $\rho_m>1$ such that
$E_{\rho_m}\subset\{|\Im z|<\delta/L_m\}$ with $\rho_m=1+c\,\delta/L_m$ yields
\[
E_m(g_m) \le A\, B_m\, \exp\!\Big(-c\,\frac{m}{L_m}\Big),
\]
and the lower bound of Theorem~\ref{thm:analytic-barrier} holds with this optimized residual.
\end{remark}

\subsection{Extension to Gevrey activation functions}

Let $s\ge1$. A function $\psi\in C^\infty([-1,1])$ belongs to the Gevrey class $G^s([-1,1])$
if there exist constants $C_\psi,R_\psi>0$ such that
\[
\|\psi^{(n)}\|_{L^\infty([-1,1])}
\le C_\psi\, R_\psi^{\,n}\,(n!)^s,
\qquad n\ge0.
\]

\begin{proposition}[Polynomial approximation of Gevrey functions]
\label{prop:gevrey-poly}
If $h\in G^s([-1,1])$, then there exist constants $A,c>0$ such that
\[
E_m(h) \le A\,\exp\!\big(-c\, m^{1/s}\big),
\qquad m\ge1.
\]
\end{proposition}

Let $\varphi\in G^s(\mathbb R)$ and consider networks $g_m$ as above.

\begin{lemma}[Uniform Gevrey bounds for network outputs]
\label{lem:gevrey-control}
There exist constants $C_\varphi,R_\varphi>0$ such that
\[
\|g_m^{(n)}\|_{L^\infty([-1,1])}
\le C_\varphi\, B_m\, (R_\varphi L_m)^n\, (n!)^s,
\qquad n\ge0.
\]
\end{lemma}

\begin{proof}
Differentiation yields
\[
g_m^{(n)}(x)
= \sum_{k=1}^m \lambda_k \alpha_k^n \varphi^{(n)}(\alpha_k x).
\]
Using the coefficient bounds and the Gevrey estimate on $\varphi$ gives the result.
\end{proof}

\begin{theorem}[Polynomial approximation barrier: Gevrey case]
\label{thm:gevrey-barrier}
Let $\varphi\in G^s(\mathbb R)$ with $s>1$. Then there exist constants $A,c>0$ such that for all
$f\in C([-1,1])$ and all $m\ge1$,
\[
\inf_{g_m} \|f-g_m\|_{L^\infty([-1,1])}
\ge
E_m(f)
-
A\, B_m\, \exp\!\Big(-c\, (m/L_m)^{1/s}\Big).
\]
\end{theorem}

\begin{proof}
Proceed as in the analytic case:
\[
\|f-g_m\|_\infty \ge E_m(f) - E_m(g_m).
\]
Apply Proposition~\ref{prop:gevrey-poly} to $g_m$, using Lemma~\ref{lem:gevrey-control}.
\end{proof}

\begin{remark}[Extended regularization paradox]
\label{rem:regularization-paradox}
The above results show that the effectiveness of neural network approximation is governed by
the \emph{combined growth} of the output weights and the inner parameters.
In the analytic setting, the residual term in
Theorem~\ref{thm:analytic-barrier} tends to zero whenever
\[
\log B_m = o\!\left(\frac{m}{L_m}\right),
\]
in which case the network class remains confined to the same approximation regime as
classical polynomial methods on generic target functions.

More generally, under the Gevrey regularity assumptions of
Theorem~\ref{thm:gevrey-barrier} with index $s>1$, the corresponding condition becomes
\[
\log B_m = o\!\left(\left(\frac{m}{L_m}\right)^{1/s}\right),
\]
and the exponential residual term is replaced by a sub-exponential one.
In both cases, regularization mechanisms that control only the magnitude of the output
weights are insufficient to guarantee improved approximation power.
As long as the combined growth condition holds, the model class remains locked into an
approximation regime dictated by its quantitative regularity, and cannot adapt efficiently
to non-analytic or non-Gevrey features of generic target functions.
\end{remark}


\begin{remark}[Relation with Barron-type weighted constraints]
\label{rem:barron-bridge}
Barron-type approximation results are based on a weighted variation control of the form
\begin{equation}\label{eq:barron-weighted}
\sum_{k=1}^m |\lambda_k|\,\|\alpha_k\| \le C_{\mathrm B},
\end{equation}
which naturally arises when discretizing an integral representation of the target function.
This constraint differs fundamentally from the uniform bounds considered in the present work,
as it does not impose any \emph{a priori} restriction on the magnitude of the inner parameters
$\alpha_k$.

\medskip
\noindent\emph{Derivative control.}
Under~\eqref{eq:barron-weighted}, and assuming that $\varphi$ is continuously differentiable
with locally bounded derivative, the associated network function
\[
g(x)=\sum_{k=1}^m \lambda_k\,\varphi(\langle \alpha_k,x\rangle)
\]
admits a uniformly bounded gradient on $[-1,1]^d$, namely
\[
\|\nabla g\|_{L^\infty([-1,1]^d)}
\;\le\;
\sup_{|t|\le L\sqrt d}|\varphi'(t)|\,
\sum_{k=1}^m |\lambda_k|\,\|\alpha_k\|,
\qquad
L:=\max_k\|\alpha_k\|.
\]
In particular, if $\varphi'$ is bounded on $\mathbb R$, the Lipschitz constant of $g$ is
uniformly controlled by $C_{\mathrm B}$, independently of the network width.
This shows that the Barron-type weighted constraint enforces a strong geometric regularity of
the model class.

\medskip
\noindent\emph{Frequency truncation and effective analyticity.}
More generally, the weighted constraint~\eqref{eq:barron-weighted} allows one to decompose
the network function as
\[
g = g_{\le L} + r_{>L},
\qquad
g_{\le L}(x):=\sum_{\|\alpha_k\|\le L}\lambda_k\,
\varphi(\langle\alpha_k,x\rangle),
\]
for any threshold $L>0$.
The high-frequency remainder satisfies the uniform bound
\[
\|r_{>L}\|_\infty
\le
\sum_{\|\alpha_k\|>L}|\lambda_k|\,\sup|\varphi|
\le
\frac{C_{\mathrm B}}{L}\,\sup|\varphi|.
\]
On the other hand, the truncated part $g_{\le L}$ admits a holomorphic extension to a complex
neighborhood of $[-1,1]^d$ whose size scales as $1/L$.
Classical Bernstein-type estimates for multivariate analytic functions then imply that
$g_{\le L}$ can be approximated by multivariate polynomials of total degree at most $m$ with
an error decaying as $\exp(-c\,m/L)$.

\medskip
\noindent\emph{Interpretation.}
This decomposition shows that Barron-type constructions escape the polynomial approximation
barrier identified in this work by allowing the effective analytic neighborhood to shrink with
the network width through the presence of increasingly large inner parameters.
In this regime, uniform analyticity is lost, and the comparison argument with polynomial
approximation no longer applies.
From this perspective, Barron-type approximation results operate precisely at the boundary of
the analyticity-controlled regime considered here.
\end{remark}

\subsection{Extension to higher dimension}

We now extend the previous results to functions defined on the hypercube $[-1,1]^d$.
Under the same assumptions on the activation function and the coefficient constraints,
network outputs admit holomorphic extensions to suitable polyelliptic neighborhoods of the
domain.

Classical Bernstein-type inequalities for multivariate analytic functions on polyelliptic
domains yield exponential bounds on the best polynomial approximation error; see, e.g.,
\cite{BorweinErdelyi1995,Pinkus1999}.

For $f\in C([-1,1]^d)$, we denote by $E_m^{(d)}(f)$ the error of best uniform approximation of
$f$ by multivariate polynomials of total degree at most $m$.

\begin{theorem}[Polynomial approximation barrier in dimension $d$]
\label{thm:main-d}
Let $d\ge1$ and let $\varphi:\mathbb R\to\mathbb R$ be a real-analytic activation function
admitting a holomorphic extension to a complex neighborhood of the real axis.
For each $m\ge1$, let $B_m>0$ and $L_m>0$, and consider the class
\[
\mathcal N_m^{(d)}(B_m,L_m)=
\Bigl\{
g(x)=\sum_{k=1}^m \lambda_k\,\varphi(\langle \alpha_k,x\rangle):
\sum_{k=1}^m|\lambda_k|\le B_m,\ \|\alpha_k\|\le L_m
\Bigr\}.
\]
Then there exist constants $A,c>0$, depending only on $\varphi$ and $d$, such that for all
$f\in C([-1,1]^d)$ and all $m\ge1$,
\[
\inf_{g\in\mathcal N_m^{(d)}(B_m,L_m)}\|f-g\|_{L^\infty([-1,1]^d)}
\;\ge\;
E_m^{(d)}(f)
-
A\,B_m\,\exp\!\Big(-c\,\frac{m}{L_m}\Big).
\]
\end{theorem}

\begin{proof}
The proof follows the same strategy as in the one-dimensional case.
Each ridge function $x\mapsto\varphi(\langle\alpha_k,x\rangle)$ admits a holomorphic extension
to a complex neighborhood of $[-1,1]^d$ whose size scales as $1/\|\alpha_k\|$.
As a consequence, every function in $\mathcal N_m^{(d)}(B_m,L_m)$ admits a holomorphic
extension to a polyelliptic neighborhood of width proportional to $1/L_m$, with a uniformly
bounded analytic norm of order $B_m$.

Applying multivariate Bernstein-type inequalities yields an approximation error bounded by
$A\,B_m\exp(-c\,m/L_m)$.
The conclusion follows by comparison with the best polynomial approximation of $f$.
\end{proof}

\begin{remark}[Comparison with $C^k$ approximation rates]
For functions in $C^k([-1,1]^d)$ and not smoother, the error of best uniform approximation by
multivariate polynomials of total degree at most $m$ cannot decay faster than a polynomial
rate, typically of order $m^{-k/d}$.
If $\log B_m=o(m/L_m)$, Theorem~\ref{thm:main-d} shows that analytic activation networks remain
confined to these classical rates on $C^k$ target functions in dimension $d$.
\end{remark}

\begin{remark}[On the necessity of coefficient growth]
Escaping the polynomial approximation regime requires either exponential growth of the output
weights or sufficiently fast growth of the inner parameters.
This observation is consistent with the constructive approximation results of
Attali and Pag\`es~\cite{AttaliPages1997}, where the approximation of high-degree polynomials by
single-hidden-layer networks with analytic activations involves coefficients that become
arbitrarily large as the target degree increases.
\end{remark}

\section{Discussion and Perspectives}

The main result of this paper identifies a fundamental limitation of single-hidden-layer
neural networks with analytic activation functions and uniformly $\ell^1$-bounded
coefficients.
Up to an exponentially small term, the best achievable approximation error on generic
target functions is lower bounded by the error of best polynomial approximation.
This shows that analyticity of the activation function, while central to several
positive approximation results, also induces intrinsic rigidity that prevents
adaptation to non-analytic features.

Our result should be interpreted as complementary to classical Barron-type
approximation theorems.
When the target function itself belongs to a structured analytic class, neural networks
can achieve fast, dimension-independent approximation rates.
In contrast, when no such regularity is assumed on the target, the analytic structure
of the model class becomes a limitation rather than an advantage.
In this sense, analyticity does not provide a universal mechanism to bypass minimax
approximation barriers on generic smoothness classes.

\begin{remark}[Relation with the constructions of Attali--Pagès]
In \cite{AttaliPages1997}, fast (in fact arbitrarily fast) approximation of polynomials by
single-hidden-layer neural networks is achieved through a singular parameter regime.
More precisely, the inner parameters scale as $\alpha_k = O(h)$ while the output
coefficients behave as $\lambda_k = O(h^{-p})$ when approximating polynomials of degree $p$.
As a consequence, the Barron-type quantity
$\sum_k |\lambda_k|\,\|\alpha_k\|$ diverges as $h \to 0$. These constructions therefore fall outside the coefficient-constrained
framework considered in the present work and do not contradict the polynomial approximation
barriers established here.
\end{remark}

The analysis developed in this work is entirely deterministic and relies on a direct
comparison between best network approximation and best polynomial approximation.
This approach differs from probabilistic constructions or random feature methods and
highlights the role of analytic continuation and Bernstein-type estimates in
understanding the expressive power of neural networks.
It provides a transparent explanation of why analytic activation functions alone cannot
improve approximation rates on non-analytic targets.

\begin{remark}[Relation with Barron-type assumptions]\label{rem:barron-summary}
Barron-type approximation results control complexity through a \emph{variation} (or weighted
$\ell^1$) norm on an integral representation of the network.
In discretized models this corresponds to $\ell^1$-type constraints on the output weights,
possibly coupled with weighted controls involving the inner parameters.
For a more detailed discussion and a decomposition explaining how Barron-type constructions
escape the uniform analyticity regime considered here, see Remark~\ref{rem:barron-bridge}.
\end{remark}

\begin{remark}[Sharpness of the lower bound]
The lower bound established in Theorem~4.3 is sharp for a large class of target functions.
Indeed, whenever the best polynomial approximation error satisfies
\[
E_m(f) \gg r_m,
\]
where $r_m$ denotes the residual term appearing in Theorem~4.3 (exponentially small in the
analytic case, or sub-exponentially small in the Gevrey case), one has
\[
\inf_{g_m}\|f-g_m\|_{L^\infty([-1,1])}
\;=\;
E_m(f)\,(1+o(1)) \qquad (m\to\infty).
\]
This includes, in particular, all non-analytic functions in classical smoothness classes
$C^k([-1,1])$ for which $E_m(f)$ decays at most polynomially.
A canonical example is $f(x)=|x|$, for which it is well known that
\[
E_m(f)\asymp m^{-1}.
\]
For such functions, the residual term is negligible and networks with analytic
activation functions under
the coefficient constraints considered here cannot outperform polynomial approximation.
\end{remark}

A natural question is whether a similar rigidity phenomenon persists when the activation
function is no longer analytic but still infinitely differentiable.
In general, such an extension is not possible without additional assumptions, since
$C^\infty$ functions need not admit any holomorphic extension and may exhibit arbitrarily
slow rates of polynomial approximation.
As a consequence, no uniform approximation barrier comparable to the analytic case can
hold for general $C^\infty$ activation functions.

Nevertheless, analogous rigidity effects can be established under stronger quantitative
regularity assumptions, such as Gevrey-type bounds on the growth of derivatives.
In that case, the approximation barrier reflects the corresponding regularity class and
leads to sub-exponential, but still non-adaptive, approximation rates.
This highlights that the phenomenon identified in this work is intrinsically tied to the
analytic structure of the activation function.

Several extensions of the present work can be considered.
First, the comparison argument extends naturally to higher-dimensional settings, where
Bernstein-type estimates for multivariate analytic functions yield analogous polynomial
barriers.
Second, it would be of interest to investigate whether similar limitations persist for
deeper architectures with analytic activations under global coefficient constraints.
Finally, the present results suggest that improved approximation on generic targets
requires either non-analytic activation functions or adaptive mechanisms that allow the
model to escape global analytic regularity constraints.

Overall, our findings contribute to a clearer understanding of the interplay between
activation regularity, model constraints, and approximation power, and help delineate
the precise scope of applicability of analytic neural network models in approximation
theory.

\bibliography{approximation}

\end{document}